\documentclass[11pt,reqno]{amsart}
\usepackage{amscd,amsmath,amsopn,amssymb,amsthm,multicol}
\usepackage{tikz,subdepth,anysize,verbatim,ifthen,xargs,colortbl,float}
\usepackage{longtable,mathtools,hyperref}
\usepackage[english]{babel}
\everymath=\expandafter{\the\everymath\displaystyle}

\theoremstyle{plain}
\newtheorem{theorem}{Theorem}
\newtheorem{prop}{Proposition}
\newtheorem{lemma}{Lemma}

\newtheorem{rk}{Remark}
\theoremstyle{definition}

\newcommand\com[1]{}

\newcommand\E{\mathcal{E}}

\newcommand\op[1]{\mathop{\rm #1}\nolimits}
\newcommand\p{\partial}

\newcommand\R{{\mathbb R}}

\begin{document}

\title{Conformal geodesics are not variational in higher dimensions}

\author{Boris Kruglikov}
\address{Department of Mathematics and Statistics, UiT the Arctic University of Norway, Troms\o\ 9037, Norway.
\ E-mail: {\tt boris.kruglikov@uit.no}. }

 \begin{abstract}
Variationality of the equation of conformal geodesics is an important problem in geometry with applications to general relativity.
Recently it was proven that, in three dimensions, this system of equations for un-parametrized curves is
the Euler-Lagrange equations of a certain conformally invariant functional, while the
parametrized system in three dimensions is not variational. We demonstrate that variationality
fails in higher dimensions for both parametrized and un-parametrized conformal geodesics, indicating that
variational principle may be the selection principle for the physical dimension.
 \end{abstract}

\maketitle

\section{Formulation of the Problem and the Result}\label{S1}

Physical laws are often given by the {\em principle of minimal action\/}, written mathematically 
in the variational form $\delta_XL=0$. 
The inverse problem of the calculus of variations is to decide whether a
given system of differential equations is the Euler-Lagrange equation of a certain functional \cite{Dav,D,An}.

Recall that geodesics satisfy the variational principle both in parametrized form (action functional) and
in non-parametrized form (length functional). In this paper we restrict to first variations only, thus essentially
discussing whether a system of differential equations is a critical point of some functional.

This problem was in a focus for the so-called conformal geodesics equation, which arose in applications to the geometry
of spatial infinity 
in general relativity \cite{S,FS}, and it was discussed in mathematical terms in \cite{BE,DK}.
The parametrized equation of conformal geodesics of a (pseudo) Riemannian metric $g$ is the following
 \begin{equation}\label{BE0}
\nabla_ua=\frac{3g(u,a)}{|u|^2}a- \frac{3|a|^2}{2|u|^2}u +|u|^2P^\sharp u-2P(u,u)u,
 \end{equation}
where $u=\dot{x}$ is the velocity field along non-null curve $x:I\to M$, $a=\nabla_uu$ is the acceleration,
the norms are defined as $|u|^2=g(u,u)$, the scalar product as $v\cdot w=g(v,w)$,
and $P$ is the Schouten tensor
 $$
P=\frac1{n-2}\Bigl(\op{Ric}_g-\frac1{2(n-1)}R_g\cdot g\Bigr).
 $$
The tangential part of this equation is responsible for a choice of parametrization, while the normal
part gives the equation of un-parametrized conformal geodesic (it is reparametrization invariant):
  \begin{equation}\label{BE2}
u\times \nabla_ua=3\frac{u\cdot a}{|u|^2}u\times a-|u|^2P^\sharp u\times u.
 \end{equation}
(One may understand $\times$ as the wedge operation, but in 3D it can be identified via Hodge star of $g$
and the musical isomorphisms with the operation of vector product on $TM$.)

Recently it was proven independently in \cite{TM} and \cite{KMS} that in 3D the equation 
for un-parametrized conformal geodesics is variational; the corresponding Lagrangian is the torsion form $\tau\,ds$. 
There exists a canonical family of parametrizations of conformal geodesics, so-called projective parameter \cite{BE}.
Whatever parameter is used, the equation of parametrized conformal geodesics in 3D is however not variational \cite{KMS}.

The higher dimensional case has remained open. One could hope that a higher dimensional version of the torsion functional 
(or any other conformally invariant Lagrangian) could solve the problem. That was verified in negative in \cite{KSS} 
in the flat 4D case, using an invariant version of the variational bicomplex. The purpose of this note is to demonstrate 
the negative answer to this problem in full generality.

 \begin{theorem}
Let $(M^n,[g])$ be a conformal manifold. 
If its dimension $n>3$, then
the equation for conformal geodesics is not variational in the classical sense neither in parametrized nor 
in unparametrized form.
 \end{theorem}

This may come as a surprise after a positive resolution of the problem in 3D \cite{TM,KMS} as well as a solution in any dimension with non-standard boundary conditions and specific class of variations \cite{DK}. 
Yet we see that in general the inverse problem of the calculus of variations in the classical setup has 
the negative solution. It is remarkable that the positive solution corresponds
to dimension 3, which has a direct physical interpretation.

\section{Warm-up: variational calculus and circles in three and two dimensions}

Let us consider curves in the Euclidean space $\R^3(x,y,z)$. Unparametrized (generic) curves are given by equation of the type $y=y(x)$, $z=z(x)$,
where $x$ can be seen as (technical) parameter, but which can be reparametrized. The curves have natural lifts to the jet-space
$J^3(\R,\R^2)=\R^9(x,y_0,z_0,y_1,z_1,y_2,z_2,y_3,z_3)$ and the equation for conformal geodesics (circles) 
determines a submanifold in this space
 \begin{equation}\label{CC3}
y_3=\frac{3y_2(y_1y_2+z_1z_2)}{1+y_1^2+z_1^2},\ z_3=\frac{3z_2(y_1y_2+z_1z_2)}{1+y_1^2+z_1^2}.
 \end{equation}
(In invariant terms this is $\kappa_s=0,\tau=0$ for the curvature $\kappa$, torsion $\tau$ and natural parameter $s$.)

For a Lagrangian $L=L(x,y_0,z_0,y_1,z_1,y_2,z_2,\dots)$ the Euler ($y$- and $z$-) operators are
 \begin{equation}\label{EL}
\frac{\delta L}{\delta y}=\frac{\p L}{\p y_0}-\frac{d}{dx}\Bigl(\frac{\p L}{\p y_1}\Bigr)+\frac{d^2}{dx^2}\Bigl(\frac{\p L}{\p y_2}\Bigr)-\dots,\qquad
\frac{\delta L}{\delta z}=\frac{\p L}{\p z_0}-\frac{d}{dx}\Bigl(\frac{\p L}{\p z_1}\Bigr)+\frac{d^2}{dx^2}\Bigl(\frac{\p L}{\p z_2}\Bigr)-\dots.
 \end{equation}
The Euler-Lagrange equations EL are
 \begin{equation}\label{ELeq}
\frac{\delta L}{\delta y}=0,\quad \frac{\delta L}{\delta z}=0.
 \end{equation}
If $L$ is a function on $J^2$ (second order Lagrangian) then no more terms appear in \eqref{EL}, so the 
Euler-Lagrange equation is of order $\leq4$, with equality for nondegenerate $L$. Equation \eqref{CC3} is 
however of order 3.

The inverse variational problem is to detect if the given equation, in our case \eqref{CC3}, is of form \eqref{ELeq} for some $L$.
Let us stress that we do not require that the first equation of \eqref{CC3} is identical with the first Euler operator in \eqref{EL}, 
etc, but the equations have to coincide as submanifolds in jets (inverse problem for differential operators is simpler 
and is solved via the so-called Helmholtz criterion \cite{An}).

For unparametrized curves the Lagrangian form $L\,dx$ should be reparametrization invariant, 
in particular $L$ does not explicitly depend on $x$. One can also assume, without loss of generality, 
that $L$ is independent of $y_0,z_0$, but this will not be important for us.

The Lagrangian for variational equation of order $k$ can be chosen to be of order $<k$: this is achieved
by modification of $L$ via divergence terms (integration by parts) and is known as Vainberg--Tonti method,
see e.g.\ \cite{Krp}. If $L$ is of order 2 then the coefficients of $y_4,z_4$ in EL are
$\frac{\p^2 L}{\p y_2^2}$, $\frac{\p^2 L}{\p y_2\p z_2}$, $\frac{\p^2 L}{\p z_2^2}$.
Since \eqref{CC3} is of the third order, these expressions have to vanish and we get

 $$
L=f(y_0,z_0,y_1,z_1)+h(y_0,z_0,y_1,z_1)y_2+r(y_0,z_0,y_1,z_1)z_2.
 $$
Using the fact, that the Lagrangian is defined up to divergence, i.e.\ that the transformation $L\mapsto L+\frac{d}{dx}f$
does not change the Euler-Lagrange equations, we can eliminate the coefficient $r\mapsto0$. Then we set our equations \eqref{CC3} into EL and split it by 2-jets $y_2,z_2$, receiving the following overdetermined PDE system:
 \begin{gather*}
y_1^2 h_{y_0,y_0} + 2y_1z_1 h_{y_0,z_0} + z_1^2 h_{z_0,z_0} =y_1 f_{y_0,y_1}+z_1 f_{z_0,y_1}-f_{y_0}, \
y_1 f_{y_0,z_1} +z_1 f_{z_0,z_1}= f_{z_0}, \\
f_{y_1,z_1} = h_{z_0}, \
f_{z_1,z_1} = 0, \
y_1 h_{y_0,z_1} +z_1 h_{z_0,z_1} =0, \
y_1 h_{y_0,y_1}+z_1 h_{z_0,y_1} =  f_{y_1,y_1} - 2h_{y_0}, \\
(1+y_1^2+z_1^2)h_{y_1,z_1} = -3y_1 h_{z_1}, \
(1+y_1^2+z_1^2) h_{z_1,z_1} = -3z_1 h_{z_1}.
 \end{gather*}
 
This system can be solved, and eliminating the remaining divergence freedom we get
 $$
L=\frac{z_1y_2}{(1+y_1^2)\sqrt{1+y_1^2+z_1^2}}.
 $$
This Lagrangian slightly simplifies the one found in formula (11) of \cite{KMS}.
Prior to this a third order Lagrangian for conformal circles in Euclidean space was established in \cite{BF1}.
Note that only this latter Lagrangian (torsion) 
is reparametrization-invariant.
The others have this property modulo divergence.

\smallskip

The inverse variational problem for parametrized curves is set up similarly in the formalism of $J^2(\R,\R^3)=\R^{10}(t,x_0,y_0,z_0,x_1,y_1,z_1,x_2,y_2,z_2)$ but the corresponding problem is resolved in negative, see \cite{KMS}.

\medskip

We can also explore a two-dimensional analog of \eqref{CC3} though in dimension 2 it should not be called the equation of
conformal geodesics (because conformal structures are uniformized; their analog in 2D should be rather called a M\"obius structure).
The unparametrized version has the form
 \begin{equation}\label{CC2}
y_3=\frac{3y_1y_2^2}{1+y_1^2}.
 \end{equation}
(In invariant terms this is $\kappa_s=0$ for the curvature $\kappa$ and natural parameter $s$, i.e.\ $\kappa=\op{const}$.)

If $L=L(x,y_0,y_1,y_2)$ then the coefficient of $y_4$ in EL is $\frac{\p^2 L}{\p y_2^2}$, so $L=u(x,y_0,y_1)+v(x,y_0,y_1)y_2$
but the Euler-Lagrange equations for such $L$ have order 2. Thus \eqref{CC2} is not variational.

\smallskip

Considering the equation for parametrized circles (see also Remark \ref{rk1} later)
 $$
x_3 = -(x_2^2+y_2^2)x_1,\quad y_3 = -(x_2^2+y_2^2)y_1,
 $$
or in conformal/M\"obius invariant way (projective parameter)
 $$
x_3=3\frac{x_1x_2+y_1y_2}{x_1^2+y_1^2}x_2 - \frac32\frac{x_2^2+y_2^2}{x_1^2+y_1^2}x_1,\quad
y_3=3\frac{x_1x_2+y_1y_2}{x_1^2+y_1^2}y_2 - \frac32\frac{x_2^2+y_2^2}{x_1^2+y_1^2}y_1.
 $$
We start with the Lagrangian $L=L(t,x_0,y_0,x_1,y_1,x_2,y_2)$ and then the condition that EL is of order $<4$ 
and simplification modulo divergences implies that we can take $L=f(t,x_0,y_0,x_1,y_1)+h(t,x_0,y_0,x_1,y_1)y_2$. 
However one of the EL equations is $h_{x_1}=0$, so $h=h(t,x_0,y_0,y_1)$ and hence a change by a divergences can bring 
the Lagrangian to the form $L=f(t,x_0,y_0,x_1,y_1)$, for which the Euler-Lagrange equation is of order at most 2. This proves 
that in dimension 2 the analog equation for (conformal) circles (both in parametric and unparametric form) is not variational.

\section{The inverse problem for conformal circles 
in four and five dimensions}\label{S4D}

Let us start with the equation for parametrized conformal geodesics. In the flat 4D space 
$\R^4$ with coordinates $x^1=x$, $x^2=y$, $x^3=z$, $x^4=w$ formula \eqref{BE0} has the following form
 \begin{equation}\label{BE0D4}
\dddot{x}^i=3\frac{\langle u,a\rangle}{|u|^2}\ddot{x}^i-\frac32\frac{|a|^2}{|u|^2}\dot{x}^i,\quad 1\leq i\leq4,
 \end{equation}
where dot denotes derivative by the parameter $t$ along the curve and 
$|u|^2=\sum_1^4(\dot{x}^i)^2$, $\langle u,a\rangle=\sum_1^4\dot{x}^i\ddot{x}^i$, $|a|^2=\sum_1^4(\ddot{x}^i)^2$.
We will also denote $x_0=x$, $x_1=\dot{x}$, $x_2=\ddot{x}$, $x_3=\dddot{x}$, and similar for $y,z,w$.

Using the Vainberg--Tonti method, we can choose 
$L=L\bigl(x_0^j,x_1^j,x_2^j:1\leq j\leq 4\bigr)$; the general Lagrangian may also depend on $t$ but this
does not change the argument below and so will be omitted.
Denote the Euler-Lagrange operator 
 \begin{equation}\label{VT}
EL_j=\frac{\delta L}{\delta x^j}=\frac{\p L}{\p x^j}-\frac{d}{dt}\frac{\p L}{\p x^j_1}+\frac{d^2}{dt^2}\frac{\p L}{\p x^j_2}.
 \end{equation}
Variationality means that $EL_j=0$ is equivalent to \eqref{BE0D4}.
The coefficients of $x^k_4$ in $\op{EL}_j$ are $\frac{\p^2f}{\p x^a_2\p x^b_2}=0$, whence
 \begin{equation}\label{L4a}
L=f_0\bigl(x_0^j,x_1^j:1\leq j\leq 4\bigr)+\sum_{k=1}^4f_k\bigl(x_0^j,x_1^j:1\leq j\leq 4\bigr)x^k_2.
 \end{equation}
This Lagrangian is degenerate in the standard sense, however our goal is to check whether the Euler-Lagrange equations are equivalent
to third order system \eqref{BE0D4}, so we call $L$ nondegenerate if the Jacobian $\Bigl[\p EL_j/\p x^k_3\Bigr]$ is invertible.
The determinant of this matrix is the square of the following quadratic expression
 \begin{multline}\label{ndg4D}
f_{1;3}f_{2;4}-f_{1;4}f_{2;3}+f_{1;4}f_{3;2}-f_{1;2}f_{3;4}+f_{1;2}f_{4;3}-f_{1;3}f_{4;2}\\
+f_{2;1}f_{3;4}-f_{2;4}f_{3;1}+f_{2;3}f_{4;1}-f_{2;1}f_{4;3}+f_{3;1}f_{4;2}-f_{3;2}f_{4;1}
 \end{multline}
which thus should be nonzero. Here $f_{i;j}$ denotes the partial derivative of $f_i$ by $x^j_1$.
If we freeze 0-jet variables $x^j_0$ getting $df_i=\sum_{j=1}^4f_{i;j}dx^j_1$, then
the nonvanishing of expression \eqref{ndg4D} is equivalent to nonvanishing of the following 4-form, 
where the indices $ijkl$ in the summation run over even permutations of $1234$:
 \begin{equation}\label{dfdx}
\sum df_i\wedge df_j\wedge dx_1^k\wedge dx_1^l\neq0.
 \end{equation}
 
Assuming this is the case, variationality of the equation $\mathcal{E}$ given by \eqref{BE0D4} is equivalent to the condition that $EL_j$ evaluated on 
\eqref{BE0D4} are identically zero. In what follows we denote this latter operation by $\op{ev}_{\mathcal{E}}EL_j$.

We now derive consequences of this constraint. Let us first modify $L$ in \eqref{L4a} by a divergence to eliminate the term 
$f_4\bigl(x_0^j,x_1^j:1\leq j\leq 4\bigr)x^4_2$. For this choose function $\psi$ such that $\p_{w_2}\psi=f_4$ and change $L\mapsto L-\frac{d}{dt}\psi$.
We obtain $L$ of form \eqref{L4a} but with summation from $k=1$ to $k=3$.

Splitting $\op{ev}_{\mathcal{E}}EL_j$ by $x^k_2$ we obtain many conditions, among which the following are 
the simplest (we now use the alternative notations for $x^j_i$ as in the beginning of this section):
 \begin{gather*}
\frac{\p^2f_0}{\p w_1\p w_1}=0,\quad 
3w_1\frac{\p f_1}{\p w_1}+(x_1^2+y_1^2+z_1^2+w_1^2)\frac{\p^2 f_1}{\p w_1\p w_1}=0,\quad
3x_1\frac{\p f_1}{\p w_1}+(x_1^2+y_1^2+z_1^2+w_1^2)\frac{\p^2 f_1}{\p x_1\p w_1}=0,
 \end{gather*}
and similarly for $f_2,f_3$.
This implies the following form of the Lagrangian where we let ${\bf v}_0=(x_0,y_0,z_0,w_0)$:
 \begin{multline*}
L= f_{00}({\bf v}_0,x_1,y_1,z_1)+f_{01}({\bf v}_0,x_1,y_1,z_1)w_1+f_{10}({\bf v}_0,x_1,y_1,z_1)x_2
+\frac{f_{11}({\bf v}_0,y_1,z_1)w_1x_2}{(x_1^2+y_1^2+z_1^2)\sqrt{x_1^2+y_1^2+z_1^2+w_1^2}}\\
+f_{20}({\bf v}_0,x_1,y_1,z_1)y_2
+\frac{f_{21}({\bf v}_0,x_1,z_1)w_1y_2}{(x_1^2+y_1^2+z_1^2)\sqrt{x_1^2+y_1^2+z_1^2+w_1^2}}\\
+f_{30}({\bf v}_0,x_1,y_1,z_1)z_2
+\frac{f_{31}({\bf v}_0,x_1,y_1)w_1z_2}{(x_1^2+y_1^2+z_1^2)\sqrt{x_1^2+y_1^2+z_1^2+w_1^2}}.
 \end{multline*}
 
Next, splitting $\op{ev}_{\mathcal{E}}EL_j$ by $x^k_2$ and $w_1$ we obtain conditions, among which 
the following are the simplest:
 \[
\frac{\p f_{01}}{\p x_1}=\frac{\p f_{10}}{\p w_0},\quad 
\frac{\p f_{01}}{\p y_1}=\frac{\p f_{20}}{\p w_0},\quad 
\frac{\p f_{01}}{\p z_1}=\frac{\p f_{30}}{\p w_0},\quad
\frac{\p f_{10}}{\p y_1}=\frac{\p f_{20}}{\p x_1},\quad 
\frac{\p f_{10}}{\p z_1}=\frac{\p f_{30}}{\p x_1},\quad 
\frac{\p f_{20}}{\p z_1}=\frac{\p f_{30}}{\p y_1},
 \]
 \begin{gather*}
\frac{\p f_{11}}{\p w_0}=\frac{\p f_{21}}{\p w_0}=\frac{\p f_{31}}{\p w_0}=0,\quad
\frac{\p f_{11}}{\p y_1}+\frac{\p f_{21}}{\p x_1}=0,\quad 
\frac{\p f_{11}}{\p z_1}+\frac{\p f_{31}}{\p x_1}=0,\quad 
\frac{\p f_{21}}{\p z_1}+\frac{\p f_{31}}{\p y_1}=0,\\
x_1f_{11}+y_1f_{21}+z_1f_{31}=0.
 \end{gather*}
From these conditions we further simplify the Lagrangian as follows
 \begin{multline*}
L= f_{00}({\bf v}_0,x_1,y_1,z_1)+\p_{w_0}\psi({\bf v}_0,x_1,y_1,z_1)w_1
+\p_{x_1}\psi({\bf v}_0,x_1,y_1,z_1)x_2\\
+\p_{y_1}\psi({\bf v}_0,x_1,y_1,z_1)y_2
+\p_{z_1}\psi({\bf v}_0,x_1,y_1,z_1)z_2
+\frac{h_3(x_0,y_0,z_0)y_1-h_2(x_0,y_0,z_0)z_1}{(x_1^2+y_1^2+z_1^2)\sqrt{x_1^2+y_1^2+z_1^2+u_1^2}} w_1x_2\\
+\frac{h_1(x_0,y_0,z_0)z_1-h_3(x_0,y_0,z_0)x_1}{(x_1^2+y_1^2+z_1^2)\sqrt{x_1^2+y_1^2+z_1^2+u_1^2}} w_1y_2
+\frac{h_2(x_0,y_0,z_0)x_1-h_1(x_0,y_0,z_0)y_1}{(x_1^2+y_1^2+z_1^2)\sqrt{x_1^2+y_1^2+z_1^2+u_1^2}} w_1z_2.
 \end{multline*}
Changing this by divergence $L\mapsto L-\frac{d}{dt}\psi$ we once more simplify the Lagrangian as follows:
 \begin{multline*}
L= h_0({\bf v}_0,x_1,y_1,z_1)
+\frac{h_3(x_0,y_0,z_0)y_1-h_2(x_0,y_0,z_0)z_1}{(x_1^2+y_1^2+z_1^2)\sqrt{x_1^2+y_1^2+z_1^2+u_1^2}} w_1x_2\\
+\frac{h_1(x_0,y_0,z_0)z_1-h_3(x_0,y_0,z_0)x_1}{(x_1^2+y_1^2+z_1^2)\sqrt{x_1^2+y_1^2+z_1^2+u_1^2}} w_1y_2
+\frac{h_2(x_0,y_0,z_0)x_1-h_1(x_0,y_0,z_0)y_1}{(x_1^2+y_1^2+z_1^2)\sqrt{x_1^2+y_1^2+z_1^2+u_1^2}} w_1z_2.
 \end{multline*}
 
Now again splitting $\op{ev}_{\mathcal{E}}EL_j$ by $x^k_2$ and $w_1$ we observe the following condition
 $$
x_1h_1+y_1h_2+z_2h_3=0, 
 $$ 
which yields $h_1=h_2=h_3=0$. Thus $L=h_0(x_0,y_0,z_0,u_0,x_1,y_1,z_1)$ is of first order,
so the corresponding Euler-Lagrange equation cannot give ODEs of the third order.
We conclude that the parametrized conformal geodesics in conformally flat 4D space are not variational.

 \begin{rk}\label{rk1}
Solutions of the parametrized conformal circles equation in any dimensions \eqref{BE0}
are obtained by rotating the motion plane to the first coordinate plane, where we get rational projective parametrization
 $$
x(t)= p+\frac{\alpha (t-r)+\beta\gamma}{(t-r)^2+\gamma^2},\quad y(t) = q+\frac{\alpha\gamma-\beta(t-r)}{(t-r)^2+\gamma^2}.  
 $$
Note that this projective parameter is the solution to a third order ODE: the scalar product of \eqref{BE0} with $u$
gives $(u,\dot{a})=3|u|^{-2}(u,a)^2-\tfrac32|a|^2$.

\smallskip
 
Another parametrization (in the conformally flat space given by the metric of nonzero constant curvature; see e.g.\ \cite{KMS}) 
is given by the equation
 $$
\nabla_ua=-|a|^2u 
 $$
for the constrained problem $|u|^2=1$, $u\cdot a=0$. In this case the solution is obtained by rotating the
motion plane to the first coordinate plane, where we get periodic parametrization
 \begin{equation}\label{abcd}
x(t)= \alpha+\gamma^{-1}\cos\bigl(\gamma(t-r)\bigr),\quad y(t) = \beta+\gamma^{-1}\sin\bigl(\gamma(t-r)\bigr).  
 \end{equation}
This parametrized equation is also non-variational, as can be checked by the same method as above.
 \end{rk}
 
Now consider un-parametrized conformal geodesic equation \eqref{BE2} in 4D. 
For this the Lagrangian form $L\,dt$ should be reparametrization invariant, so we choose the 
coordinate $w$ as local time, and $x^1=x$, $x^2=y$, $x^3=z$ will depend on it.
Thus $L=L(x_i,y_i,z_i)$ where the derivatives are taken wrt $w$. 
As before, the Lagrangian can be chosen of order $2$, i.e.\ $0\leq i\leq2$ in the expression for $L$,
and it should be quasi-linear, i.e.\ $L=f_0+\sum_{k=1}^3f_kx^k_2$, where $f_i=f_i(x_0,y_0,z_0,x_1,y_1,z_1)$.
This implies for $1\leq j,k\leq3$
 $$
\frac{\p EL_j}{\p x_2^k}=\frac{\p f_j}{\p x^k_1}-\frac{\p f_k}{\p x^j_1}.
 $$
Thus the $3\times3$ symbol matrix of the Euler-Lagrange equation is skew-symmetric and hence degenerate.
Therefore it cannot reproduce the un-parametrized conformal geodesic equation.

 \begin{rk}
The last argument is borrowed from the proof of Proposition 2 from \cite{KMS}. In a similar vein,
the same argument implies, right away, the non-variationality for parametrized conformal geodesics 
in odd dimensions and for un-parametrized conformal geodesics in even dimensions.
 \end{rk}
 
In 5D the equation for parametrized conformal geodesics is trivially non-variational, as follows from the
above remark. Un-parametrized conformal geodesics are given by the following system
(where the bracket $[\cdot\cdot]$ stands for skew-symmetrization over the indicated indices)
 \begin{equation}\label{BE2D5}
\dot{x}^{[i}\dddot{x}^{j]}=3\frac{\langle u,a\rangle}{|u|^2}\dot{x}^{[i}\ddot{x}^{j]},\quad 
0\leq i<j\leq4.
 \end{equation}
Choosing $x^0=t$ as parameter we get the system (the indices are lowered via the Euclidean metric)
 \begin{equation}\label{BE3D4}
\dddot{x}^j=\frac{3\dot{x}_i\ddot{x}^i}{1+\dot{x}_i\dot{x}^i}\ddot{x}^j,\quad 
1\leq j\leq4.
 \end{equation}
Applying the  same method as for parametrized conformal geodesics in 4D, one can verify that this 
equation is also non-variational. These tedious computations are fully similar, so we omit the details.

\section{Reduction and Proof in the general case}\label{SgD}

Now we are ready to prove the main theorem. The idea is to reduce the claim
to the simplest cases, considered before, both in dimension and symmetry-wise. The proof will be split in steps.

\medskip
 
{\bf 0. General signature.} In the computations above and below we use Riemannian signature of $g$.
When metric $g$ is analytic, this is equivalent to the claim on variationality for any signature of the conformal structure
$[g]$ (recall that the conformal geodesics are non-null curves $|u|^2=g(u,u)\neq0$).
It is easy to see that if in one domain of sign choices the curves are locally variational,
then they are such in complexification and hence (Wick rotation) variationality holds
in other real versions of $g$ and other choices of sign for $|u|^2$.

Actually, one does not need analyticity of $g$ in this argument because the variationality is decided on the level 
of finite jets, hence one can work formally. Thus we can assume $g$ smooth, without loss of generality.

Indeed, the problem is to find solvability of the system $\delta_XL=0$ restricted to the equation $\mathcal{E}$
given by \eqref{BE0} or \eqref{BE2}, respectively, for a second order Lagrangian affine in 2-jets $x_2^i$, 
together with the non-degeneracy condition 
 \begin{equation}\label{ndg}
\det\Bigl[\frac{\p EL_j}{\p x^i_2}\Bigr]\neq0
 \end{equation}
that has differential-algebraic form, cf.\ \eqref{ndg4D}.
Splitting by jets the equation $\op{ev}_\mathcal{E}EL_j$ we get an overdetermined system, which is algebraic in jets
and hence allows complexification. Involutivity of this system with relation \eqref{ndg} is an algebraic condition,
not involving any real terms.

\medskip

{\bf 1. Parametrized flat structures in any dimension.} We claim that the equation for parametrized conformal 
geodesics in the Euclidean space of dimension $n>3$ is non-variational. 
For $n=3$ this was demonstrated in \cite{KMS},
for $n=4$ was proven in Section \ref{S4D} and also confirmed by invariant calculus in \cite{KSS}. 

We proceed by induction. Note that hyperplanes in $\R^n$ are totally conformal geodesic:
every conformal geodesic starting on the point of the hyperplane $\Pi\subset\R^n$ with velocity and acceleration 
in $\Pi$ remain in $\Pi$ for all times. This is easy to see as the curves are just projectively parametrized circles.
Also note that substitution of $x^n_k=0$ for $0\leq k\leq 3$ into \eqref{BE0} reduces the conformal
circle equation in $\R^n$ to that in $\R^{n-1}$.

Using a rotation $R\in SO(n)$ and a translation $T\in\R^n$ we represent $\Pi=RT(\Pi_0)$ for $\Pi_0=\{x^n=0\}$.
Assume, on the contrary, that the conformal geodesic equation in $\R^n$ is variational.
 \begin{lemma}
Let $L_0=L|_{x^n_k=0}$ be the restriction. Then the Euler-Lagrange equations $\frac{\delta L}{\delta x^j}=0$
for $1\leq j<n$ restricted to $x^n_k=0$ coincide with the Euler-Lagrange equations for $L_0$.
 \end{lemma}

 \begin{proof}
From \eqref{VT} we see that the restriction of $EL_j$ for $j<n$ to the constraint $\{x^n_k=0:\forall k\geq0\}$
is equivalent to the corresponding Euler-Lagrange equation with restricted $L$: this follows from 
$\tfrac{d}{dt}$ closedness of the ideal describing the constraint.
 \end{proof}

Note however that nondegeneracy condition \eqref{ndg} for $L$ fails in odd dimensions.
Indeed, the determinant in \eqref{ndg} is the square of the Pfaffian of the skew-symmetric matrix
$\bigl[\tfrac{\p EL_j}{\p x^i_2}\bigr]$ and in odd dimensions this always vanishes. 
Thus one can make a reduction by two dimensions, when restriction yields nonzero Pfaffian,
and complete the induction step based on non-variationality of the parametrized conformal circles
in dimensions $n=3,4$ (in fact, one could start at $n=2$ but this is not conformal but M\"obius geometry):

 \begin{prop}
The equation for parametrized conformal circles in $(\R^n\!,ds^2_{\op{Eucl}})$ is not variational for $n\ge3$.
 \end{prop}

\medskip

{\bf 2. Un-parametrized flat structures in any dimension.} 
Here the reduction is very similar, and we again use totally conformal geodesic property.
This gives a reduction of dimensions in the flat Euclidean case. An essential difference is that
the nondegeneracy condition
 \begin{equation}\label{open}
\op{Pf}\Bigl[\frac{\p EL_j}{\p x^i_2}\Bigr]\neq0
 \end{equation}
generally holds in odd dimensions and always fails in even dimensions. 
To get proper induction one should get base of induction for $n=4,5$, which was resolved in Section \ref{S4D}. 

Since investigation of the equation for un-parametrized conformal circles in 5D was only briefly discussed there,
let us present an important modification. Nondegeneracy condition \eqref{open} signifies equivalence
of the Euler-Lagrange equations to \eqref{BE2}. Yet, one may consider overdetermined system 
$\op{ev}_\mathcal{E}EL_j$ split by $x^i_2$ even without this inequality. 
An important step in simplification was change by divergence $L\mapsto L+\tfrac{d}{dt}\psi$,
but this can be done also when the order of $L$ is below 3, that is, without restriction to $\mathcal{E}$.
In this case the Euler-Lagrange equation holds iff $L$ is a pure divergence, so it reduces to 0.

This has been directly verified for $n=5$. Thus it finishes the induction step and proves
 \begin{prop}
\!The equation  for unparametrized conformal circles in $(\R^n\!\!,ds^2_{\op{Eucl}})$ is not variational for $n>3$.\!\!
 \end{prop}

Let us make a remark on the above reduction procedure. We prove non-variationality via reduction of dimension by
two in both cases: for parametrized conformal geodesics dimensions drop $2n\mapsto 2n-2\mapsto\dots$,
for un-parametrized conformal geodesics dimensions drop $2n+1\mapsto 2n-1\mapsto\dots$ to keep nondegeneracy.
One may expect a dimension reduction by one through switching between the two cases. This is indeed possible for
the corresponding equations: for parametrized curves forget about parametrization and then restrict to a hyperplane;
for unparametrized curves choose the projective parametrer and then restrict to a hyperplane. However
in terms of Lagrangians the cases are distinguished by the property of the action $\int L\,dt$ 
being reparemetrization invariant and this does not change under restriction (of the Lagrangian density and
variations) to the equation of hypersurface. 
Thus reduction by one dimension seems not possible.

 \bigskip

{\bf 3. General non-flat case.} 
The above results imply that conformal geodesics in general dimension are not variational, as we shall now demonstrate. 
Let us first deduce the negative results for generic metrics $g$ from computations for flat metrics $g_0$.
By generic we mean such $g$ that its $k$-jets is generic at a reference point for specified $k$.
We omit large explicit formulae for equations, referring to the simplest representatives in the flat case, as this is not
essential for our non-computational argument.

We adapt the approach from the discussion in the previous subsection for un-parametrized case, 
without using nondegeneracy condition \eqref{open}; the variationality is given by a system of closed differential-algebraic
relations on the entry of the metric $g$ representing the conformal structure. 
Since these conditions fail for flat metric $ds^2_{\op{Eucl}}$ they also fail for metrics $g$ in an open 
neighborhood in some $C^k$ topology (actually, $k=5$ suffices, but we will not use it). 
Moreover, due to algebraicity of the conditions, the variationality fails for metrics $g$
in an open set in Zariski topology on $k$-jets of the metrics, whence for generic metrics.

In fact, we can obtain obstructions to variationality without any explicit integration of particular constraints, as in 
Section \ref{S4D}. The condition that $\delta_XL=0$ (without restriction to equation $\E$ of conformal geodesics) 
is equivalent to the condition of complete divergence $L\,dt=\hat{d}\psi$, where $\hat{d}$ is the horizontal differential 
(see \cite{IA,VK} for details on this via variational bicomplex/$\mathcal{C}$-spectral sequence).

Denote by $\Delta^0$ the equation of complete divergence, and by $\Delta^\E$ the Euler-Lagrange condition
on equation that we previously denoted $\op{ev}_\mathcal{E}EL$. Both are submanifolds obtained by prolongations
in the corresponding jet spaces (cf.\ \cite{VK} for details on jets and differential equations) and nested so: 
 \begin{equation}\label{Delta}
\Delta_k^0\subset\Delta_k^\E\subset J^k.
 \end{equation}
Both $\Delta_k^0,\Delta_k^\E$ are differential equations on components $f_i\in C^\infty(J^1(M^n,1))$ of the Lagrangian 
$L=f_0+\sum_{j=1}^nf_jx_2^j$, depending on 1-jets of curves in $M$, hence the space in the right-hand-side 
of \eqref{Delta} is the space of $k$-jets of functions $f_i$ on $J^1(M^n,1)$.
Note that $\Delta^0$ expresses triviality of the Euler-Lagrange equation for the above $L$, while $\Delta^\E$
expresses this triviality modulo $\E$.

Since PDEs in \eqref{Delta} are linear, the quotient of the corresponding jet-subbundles $\Delta_k^\E/\Delta_k$ 
is a bundle over $M$ co-filtered by $k$.
While equation $\Delta_k^\E$ is of infinite type, the corresponding quotient equation on $[f_i]$ is of finite type.
It means that there are only finitely many linearly independent jets of $f_i$ in the quotient $\Delta^\E/\Delta^0$,
and we call normal jets a basis of those.
Indeed, this finite type property we have observed for flat metrics and it does not change for general metrics, 
as it only depends on the symbol of the equation.

Actually, the equation $\Delta^\E$ has order 2 in $f_i$ and also order 2 in parameter metric $g$
(it enters through Christoffel symbols, Schouten tensor, etc).
Its third prolongation ($k=3+2$) allows to express all normal jets $[j^kf_i]\in \Delta_k^\E/\Delta_k^0$.
In other words, the matrix $B_g$ of the linear system on these normal jets is of maximal rank for $g=g_0$
being the flat metric. But then it is also of maximal rank for any metric $g$ with generic $k$-jet,
whence $[j^kf_i]\in\Delta_k^0$ and this means that $L$ is a complete divergence. 
 Thus we have proved:
 
 \begin{prop}
For generic conformal metric $[g]$ the equation of conformal geodesics on $(M,[g])$,
both in parametrized or unparametrized version, is not variational.
 \end{prop}
 
Let us make a remark on the usage of a reference point (where the $k$-jet is based) at which we evaluate the rank of $B_g$
in the above proof.
For (conformally) flat metric all points play the same role due to homogenuity. In general, there can be regular and singular 
points (with respect to stability of ranks of various tensors). 
The former are generic and we may always choose the reference point regular.
However we may also take any singular point, because the argument is based on comparisons of $B_g$ with $B_{g_0}$.
 
 \begin{rk}
Alternatively the above argument can be presented as follows. The condition that a (nonlinear) differential operator $F$ is
the Euler-Lagrange operator of some $L$ is given by the self-adjointness condition
(here $\ell_F$ is the operator of universal linearization of $F$ and $\ell_F^*$ is its formal adjoint \cite{V})
 $$
\ell_F^*=\ell_F.
 $$ 
The conformal geodesic equation $\E:F=0$ is equivalent to $\Box F$ 
for an invertible $C$-differential operator $\Box$ (we refer to \cite{VK} for the formalism); 
in our case the coefficients of $\Box$ can be assumed functions on $J^1(M,1)$.
Solvability of the differential equation $\ell_{\Box F}^*=\ell_{\Box F}$ on these coefficients is a finite type
system, which is not solvable for the parameter $g_0$ and hence not solvable for a generic metric $g$.
 \end{rk}

To conclude the proof we extend this negative result for all metrics $g$ in higher dimensions. 
The maximal rank condition for the matrix $B_g$ is metric invariant, i.e.\ if $g'=\varphi^*g$ for a diffeomorphism 
$\varphi:M\to M$ then $B_g$ and $B_g{}\!'$ have the same ranks in the corresponding
jet-points. Dropping of the rank is a metric invariant condition and hence it can be expressed 
via differential invariants of $g$ (absolute and relative, cf.\ \cite{O}).

Fixing a point $x\in M$, the differential group $D_x^{k+1}$ of $(k+1)$-jets of diffeomorphisms fixing $x$ 
naturally acts on the space of jets of metrics $g$ on $M$, which is an open subset in $J^k_x(M,S^2T^*M)$.
This action of $D_x^{k+1}$ is algebraic. 
Thus its only closed orbit is the orbit of the jet at $x$ of the flat metric $g_0$ (where the stabilizer subgroup is largest) 
and all other orbits are adjacent to it. Consequently, by lower semi-continuity of the rank,
for any metrics $g$ the rank of $B_g$ cannot be smaller than that of $g_0$, and
hence for any metric $g$ a solution $f_i$ of $\Delta^\E$ is a complete divergence. 
This finishes the proof of the main theorem.

\section{Discussion: on the selection principle}\label{conclusion}

Conformal geodesics extend naturally to (and through) conformal infinity. Their behavior at the conformal infinity 
uncovers the structure of asymptotically simple solutions of the Einstein field equations \cite{S,FS}. 
In infinity this family of curves carries only a projective family of parametrizations, which makes 
un-parametrized conformal geodesics at the three-dimensional spatial boundary physically meaningful.

Our main result states that the variationality of conformal geodesics is the selection principle for physical dimension 
three. 
In all other dimensions (including dimension two and parametrized version
in dimension three) the equation is not variational. Here we test only the first variation of conformal geodesics.
Indeed, even in the flat case, there exist infinitely many conformal geodesics through any two (sufficiently close) 
points and thus the second variation does not seem to be relevant.

Let us stress that we consider the standard approach to variational problem when we test whether the given
equation $\E$ is the Euler-Lagrange equation $\delta_XL=0$ of a certain Lagrangian $L$. The equation may be
considered geometrically as a submanifold $\E\subset J^k$ of some jet-space. Sometimes, in the literature, especially 
in investigation of integrability, the equation is exchanged to its intermediate integral or differential corollary. 
In geometric terms, this means either smaller or larger equation-submanifold in the space of jets:
 $$
\E_\vartriangle\subset \E\subset \E^\triangledown.
 $$

For instance, in the case of conformal geodesic equation $\E$ a nondegenerate conformal invariant Lagrangian 
of order 2 was proposed in \cite{BE}, which up to divergence coincides with the one investigated further in \cite{DK}. 
It was noted in \cite{BE} that its Euler-Lagrange equation is of forth order; this equation 
$\E^\triangledown\supset\E$ containing the equation of conformal geodesics 
was later studied in \cite{SZ,PZ}. An example of smaller equation $\E_\vartriangle$ (for K\"ahler metric $g$) 
is given by K\"ahler magnetic geodesics discussed in \cite{DK}. 

We note that, in general, variationality of equation $\E$ cannot be exchanged with that for $\E_\vartriangle$ or 
$\E^\triangledown$. Indeed, as known from math folklore, there exists a Lagrangian $L^\triangledown$ 
whose extremals contain all solutions of $\E$:

 \begin{prop}
Any equation $\E$ can be embedded into a variational equation $\E^\triangledown$. 
For the latter one can take the cotangent covering given by the equations $\E$ and $\ell_\E^*$
(formal adjoint of the linearization operator $\ell_\E$).
 \end{prop}


 \begin{proof}
Let the equation be given by $\E=\{F_i=0\}_{i=1}^m$, where $F_i$ are (nonlinear) differential operators on
$u=(u^j)_{j=1}^n$ (usually, for determined systems, $m=n$).  
Then $L^\triangledown=\sum_{i=1}^m v^iF_i$ for additional dependent variable $v=(v^i)_{i=1}^m$.
Indeed, the equations $\delta_{v^i}L^\triangledown=F_i=0$ for $1\leq i\leq n$ imply $\E$.
The other Euler-Lagrange equations $\delta_{u^j}L^\triangledown=\sum_{i=1}^nv^i\delta_{u^j}F_i=0$
should be considered as constraints on $v$. 
Taken together these two sets of equations define what is known as cotangent covering $\mathcal{T}^*\E$ over $\E$,
cf.\ \cite{VK}. Its linearization is given by the operator $\ell_\E\oplus\ell_\E^*$ that is self-adjoint 
and hence (see e.g.\ \cite{IA,V}) $\mathcal{T}^*\E$ is a variational equation.
 \end{proof}

Similarly, any differential equation $\E$ (with more than one-dimensional solution space) contains 
a reduction to scalar second order ODE (path geometry) that is variational (a result by Sonin of 1886, see \cite{KM}). 
In other words, under mild conditions, every equation $\E$ contains a variational intermediate integral $\E_\vartriangle$.

\smallskip

Our final remark concerns variations with prescribed boundary conditions. With the standard approach
one usually assumes fixed endpoint and derivatives of variations at it. Non-standard variational principle
for conformal geodesics in \cite{DK} restricts the variations $V$ along the curve by $V'(t_0)=0$, $V'(t_1)=0$
as well as the boundary condition $B(V)|_{t_0}^{t_1}=-K(V)$, where $B$ is a certain second order differential
operator. For the flat metric $K=0$. Consider, for instance, a one-parametric family of conformal circles in $\R^n$,
with center $(0,\tfrac{1-c^2}{2c},0,\dots,0)$ and radius $\tfrac{c^2+1}{2c}$,
passing through the points $(\pm1,0,\dots,0)$ for $t=\pm1$:
 $$
x_1(t)=\frac{(c^2+1)t}{t^2+c^2},\ x_2(t)=\frac{c(1-t^2)}{t^2+c^2},\ x_3(t)=\dots x_n(t)=0.
 $$ 
Then $B(V)=\frac1{1+c^2}\Bigl((c^2-t^2)V_1''-2ctV_2''-4tV_1'-4cV_2'-2V_1\Bigr)$ and hence 
 $$
\bigl.B(V)\bigr|_{t=-1}^{t=+1}=\frac{c^2-1}{c^2+1}\bigl(V_1''(1)-V_1''(-1)\bigr)
-\frac{2c}{c^2+1}\bigl(V_2''(1)+V_2''(-1)\bigr)-\frac2{c^2+1}\bigl(V_1(1)-V_1(-1)\bigr).
 $$
Thus the allowed variations depend on a particular extremal, or more generally a curve in $M$.
Hence the boundary conditions are selective for a subset of all solutions of $\E$ in this case.

This shows some differences between the selection principle of our work with the previous publications
on variational principles for conformal geodesics.

\smallskip

 \begin{rk}
Conformal geodesics are distinguished curves for conformal geometry.
More general distinguished curves in parabolic geometries include: 
chains in CR geometry (variational \cite{CMMM}) and canonical curves in Lagrangian contact structures 
(also variational \cite{MFMZ}), geodesics in projective geometry (generically non-variational \cite{KM})
and paths for systems of second order ODEs (also generically non-variational \cite{KM}).
This raises the question which distinguished curves in parabolic geometries are variational.
 \end{rk}

Let us finish with a  historical remark. The notion of conformal geodesics has appeared in literature
in at least three different contexts. The first corresponds to the title of Fialkov's paper \cite{F}, where he mentions
the preceeding work by Schouten of 1928, and there conformal geodesics are defined as a family of curves,
more precisely a path structure that is metrizable within the given conformal class. No particular curve from
this family can be characterized by a relation to the conformal structure. Indeed, it was shown in Lemma 2 of \cite{EZ}
that any curve is locally a geodesic for a metric in a given conformal class. 

Next, conformal geodesics were introduced as extremals for the so-called conformal arclength functional
by Musso in \cite{M} and then extended in \cite{MMR} (this action corresponds to the simplest differential invariant one-form). 
These curves are variational by construction, but their Euler-Lagrange equation is of order six. 

Finally, our usage of conformal geodesics is related to \cite{FS}, which is based on the work by Schmidt of 1986 with a reference 
to the preceeding paper by Ogiue of 1967; it was re-used in physical literature by Tod in 2012 and later, see references in \cite{DK}. 
In mathematical literature this equation was introduced in \cite{Y} in 1957; later in \cite{BE} it was more conservatively called 
the equation of conformal circles. Given the main result of this paper it may be questioned which nomenclarure is more justified.

\bigskip

{\bf Acknowledgment.} I would like to thank Vladimir Matveev, Wijnand Steneker, Eugene Ferapontov, 
Iosif Krasilshchik, Konstantin Druzhkov, Michael Eastwood and Josef Silhan for useful discussions. 
The research was supported by the Tromsø Research Foundation 
(project “Pure Mathematics in Norway”) and the UiT Aurora project MASCOT.


\end{document}